\documentclass[twoside]{article}
\usepackage{qic}

\usepackage{amsmath}
\usepackage{amssymb}
\usepackage{braket}
\usepackage{color}

\newtheorem{thm}{Theorem}

\newtheorem{lem}{Lemma}

\begin{document}
\setlength{\textheight}{8.0truein}    %FOR 2ND PAGE ONWARDS

\runninghead{A limit theorem for a 3-period time-dependent quantum walk}
            {F. A. Gr{\"u}nbaum and T. Machida}

\normalsize\textlineskip
\thispagestyle{empty}
\setcounter{page}{1}

%\copyrightheading{Vol.}{No.}{Year}{Page Nos.}
%\copyrightheading{0}{0}{2003}{000--000}

\vspace*{0.88truein}

\alphfootnote

\fpage{1}

%\centerline{\bf
%%%%%%%%%%%%%%%%%%%%%%
%%Put in titiles here
%%%%%%%%%%%%%%%%%%%%%%
%\ourtitle}
\centerline{\bf
%%%%%%%%%%%%%%%%%%%%%
%Put in titiles here
%%%%%%%%%%%%%%%%%%%%%
A limit theorem for a 3-period time-dependent quantum walk}
\vspace*{0.37truein}
\centerline{\footnotesize
%%%%%%%%%%%%%%%%%%%%%%%%%%%%%%%%%%%%
%put authors' name and address here
%%%%%%%%%%%%%%%%%%%%%%%%%%%%%%%%%%%%
F. Alberto Gr{\"u}nbaum}
\vspace*{0.015truein}
\centerline{\footnotesize\it Department of Mathematics, University of California,}
\baselineskip=10pt
\centerline{\footnotesize\it Berkeley, CA, 94720, USA}
\vspace*{0.225truein}
\centerline{\footnotesize
%%%%%%%%%%%%%%%%%%%%%%%%%%%%%%%%%%%%
%put authors' name and address here
%%%%%%%%%%%%%%%%%%%%%%%%%%%%%%%%%%%%
Takuya Machida}
\vspace*{0.015truein}
\centerline{\footnotesize\it Department of Mathematics, University of California,}
\baselineskip=10pt
\centerline{\footnotesize\it Berkeley, CA, 94720, USA}
\vspace*{0.015truein}
\centerline{\footnotesize\it Research Fellow of Japan Society for the Promotion of Science,}
\baselineskip=10pt
\centerline{\footnotesize\it Meiji University, Nakano Campus, 4-21-1 Nakano, Nakano-ku, Tokyo 164-8525, Japan}
\vspace*{0.225truein}
%\publisher{(received date)}{(revised date)}

\vspace*{0.21truein}

%% \abstracts{first paragraph}{second paragraph}{third paragraph}
%% If there is only one paragraph, just keep the second and third empty 
%% like the following one 
\abstracts{
%%%%%%%%%%%%%%%%%%%%
% put abstract here
%%%%%%%%%%%%%%%%%%%%
We consider a discrete-time 2-state quantum walk on the line.
The state of the quantum walker evolves according to a rule which is determined by a coin-flip operator and a position-shift operator.
In this paper we take a 3-periodic time evolution as the rule.
For such a quantum walk, we get a limit distribution which expresses the asymptotic behavior of the walker after a long time.
The limit distribution is different from that of a time-independent quantum walk or a 2-period time-dependent quantum walk. We give some analytical results and then consider a number of variants of our model and indicate the result of simulations for these ones.
}{}{}

\vspace*{10pt}

%\vspace*{3pt}
%\communicate{to be filled by the Editorial}

\vspace*{1pt}\textlineskip    %) USE THIS MEASUREMENT WHEN THERE IS
   %) A SECTION HEADING
%\vspace*{-0.5pt}
%\noindent

\bibliographystyle{qic}

%%%%%%%%%%%%%%%%%%%%%%%%%%%%%%%%%%%   INTRODUCTION   %%%%%%%%%%%%%%%%%%%%%%%%%%%%%%%%%%%%%%%%%%%%%%%%%%%%%%%%%
\section{Introduction}
Quantum walks (QWs) are considered to be a quantum analog of classical random walks.
The system and the dynamics of QWs have some similarities to those of random walks, but the behavior of QWs is different from that of random walks in terms of their probability distributions .
In general, the behavior of the QWs can not be predicted based on our intuition.
A 3-period time-dependent QW which we are going to consider in this paper leads to an interesting behavior.
We study this behavior after a large number of discrete time steps and describe it as a long-time limit theorem.
The theorem will be given as a convergence in distribution on a rescaled space by time. The fact that the relevant scale is time itself and not its square root has been observed from the very first papers in the subject, \cite{AharonovDavidovichZagury1993}.
For a time-independent standard QW on the line, a limit distribution was obtained by Konno~\cite{Konno2002a,Konno2005} in 2002 for the first time and the limit density function has a representation similar to an arcsine law, in marked contrast to a Gauss distribution  which appears for classical random walks under appropriate conditions.
Time-dependent QWs were numerically studied in some papers~\cite{MackayBartlettStephensonSanders2002,RibeiroMilmanMosseri2004,BanulsNavarretePerezRoldanSoriano2006,Romanelli2009} and some limit theorems were analytically derived \cite{MachidaKonno2010,Machida2011,Machida2013b,IdeKonnoMachidaSegawa2011}.
In particular, Machida and Konno~\cite{MachidaKonno2010} treated a 2-period discrete-time QW on the line whose time evolution is given by two unitary matrices which are used as coin-flip operators.
The long time behavior of the 2-period time-dependent walk can be completely determined by one of the two matrices according to the determinant of the product of both of them.

In this paper we define a 3-period time-dependent discrete-time QW on the line and we will see that this 3-period time-dependent walk also exhibits interesting behavior.
The motivation for the analytical study for the 3-period time-dependent walk done here comes from numerical studies done in Ribeiro et al.~\cite{RibeiroMilmanMosseri2004}.
Besides periodic time-dependent walks, they also looked at time-dependent QWs whose coin-flip operator was controlled by a quasiperiodic sequence or a random sequence.
According to their result, we can expect that the long time behavior of a walk with a long period is sub-ballistic or diffusive.
That means that as the length of the period increases, the behavior of the periodic time-dependent walks gets either less ballistic or more diffusive departing form the behavior of a time-independent quantum walk.
So, we would see a different behavior for a periodic QW depending on the length of the period, and this would be important in order to discuss the relationship between QWs and random walks.

We will define a 3-period time-dependent QW on the line in the following section.
The walker starts from the origin on the lattice $\mathbb{Z}=\left\{0,\pm 1,\pm2,\ldots\right\}$ at time 0 and from its state at time $t\in\left\{0,1,2,\ldots\right\}$ one gets the state at time $t+1$ after operating with  a coin-flip operator and a position-shift operator.
In our model the coin operator is 3-periodic as a function of time $t$, and we use just one and the same coin-flip operator in the evolution.
For the 3-period time-dependent walk, we give a limit theorem as $t\to\infty$ in Sec.~\ref{sec:limit_th}.
The proof of the theorem is based on Fourier analysis and is included in the same section.
In the final section, we give a summary and a discussion of our result.

There are two appendices: in the first one we show how the analytical proof can be made to work in the case of some unitary (as opposed to orthogonal) operators.
In the second one we look at a number of models not covered by our analytical results and give some interesting numerical evidence of their limiting behavior.

%%%%%%%%%%%%%%%%%%%%%%%%%%%%%%%%%%%   DEFINITION   %%%%%%%%%%%%%%%%%%%%%%%%%%%%%%%%%%%%%%%%%%%%%%%%%%%%%%%%%
\section{Definition of a 3-period time-dependent QW on the line}
\label{sec:definition}
In this paper we deal with a discrete-time 2-state QW on the line and we give a 3-periodic time evolution rule for the walk.
The total system of a discrete-time 2-state QWs on the line is defined in a tensor space $\mathcal{H}_p\otimes\mathcal{H}_c$, where $\mathcal{H}_p$ is called a position Hilbert space which is spanned by an orthogonal normalized basis $\left\{\ket{x}:\,x\in\mathbb{Z}\right\}$ and $\mathcal{H}_c$ is called a coin Hilbert space which is spanned by an orthogonal normalized basis $\left\{\ket{0},\ket{1}\right\}$.
Let $\ket{\psi_{t}(x)} \in \mathcal{H}_c$ be the state of the walker at position $x$ at time $t$.
The state of the 2-state QW on the line at time $t$ is expressed by $\ket{\Psi_t}=\sum_{x\in\mathbb{Z}}\ket{x}\otimes\ket{\psi_{t}(x)}\in\mathcal{H}_p\otimes\mathcal{H}_c$.
In particular, we focus on a 3-period time-dependent discrete-time QW whose coin-state is given by
\begin{align}
 C=&\cos\theta\ket{0}\bra{0}+\sin\theta\ket{0}\bra{1}+\sin\theta\ket{1}\bra{0}-\cos\theta\ket{1}\bra{1}\nonumber\\
 =&c\ket{0}\bra{0}+s\ket{0}\bra{1}+s\ket{1}\bra{0}-c\ket{1}\bra{1},
 \label{eq:coin-flip operator}
\end{align}
with $\theta\in [0,2\pi)$ and we have abbreviated $\cos\theta, \sin\theta$ to $c, s$ in Eq.~(\ref{eq:coin-flip operator}).
The total system at time $t$ evolves to the next state at time $t+1$ according to the time evolution rule
\begin{equation}
 \ket{\Psi_{t+1}}=\left\{\begin{array}{ll}
		   \tilde{S}\tilde{C}\ket{\Psi_t}& (t=0,1 \mod 3)\\[1mm]
			  \tilde{S}\ket{\Psi_t}& (t=2 \mod 3)
			 \end{array}\right.,
\label{eq:time-evolution}
\end{equation}
where
\begin{align}
 \tilde{C}=&\sum_{x\in\mathbb{Z}}\ket{x}\bra{x}\otimes C,\\
 \tilde{S}=&\sum_{x\in\mathbb{Z}}\ket{x-1}\bra{x}\otimes\ket{0}\bra{0}+\ket{x+1}\bra{x}\otimes\ket{1}\bra{1}.
\end{align}
The time evolution of the state $\ket{\Psi_t}$ depends on the value $t \mod 3$.
Equation~(\ref{eq:time-evolution}) states that the position of the walker gets shifted after the coin-flip operation has been completed at time $t=0,1 \mod 3$, and it just gets shifted without any coin-flip operation at time $t=2 \mod 3$.
Here, we don't take $\theta=0,\frac{\pi}{2},\pi,\frac{3\pi}{2}$ because the behavior of the walker would be trivial.
Under the condition $\braket{\Psi_0|\Psi_0}=1$, the quantum walker can be observed at position $x$ at time $t$ with probability
\begin{equation}
 \mathbb{P}(X_t=x)=\bra{\Psi_t}\biggl\{\ket{x}\bra{x}\otimes (\ket{0}\bra{0}+\ket{1}\bra{1})\biggr\}\ket{\Psi_t},
\end{equation}
where $X_t$ is a random variable and denotes the position of the walker at time $t$, regardless of the spin orientation.
The probability distribution evolves as a function of time $t$, as numerically shown in Fig~\ref{fig:time-probability}.
Actually, this linear behavior is reflected in a limit theorem which will show up after this section.
We also show how the time evolution of the probability distribution depends on the parameter $\theta$ of the coin-flip operator $C$ in Fig.~\ref{fig:theta-probability}.
We will analyze the long time behavior of this probability distribution $\mathbb{P}(X_t=x)$ as $t\to\infty$ in the next section, concentrating on values of time that are of the form $3t$. Other values of time show an undistinguishable behavior (see also Appendix~\ref{app:3t+1_and_3t+2}).

\begin{figure}[h]
\begin{center}
 \begin{minipage}{70mm}
  \begin{center}
   \includegraphics[scale=0.7]{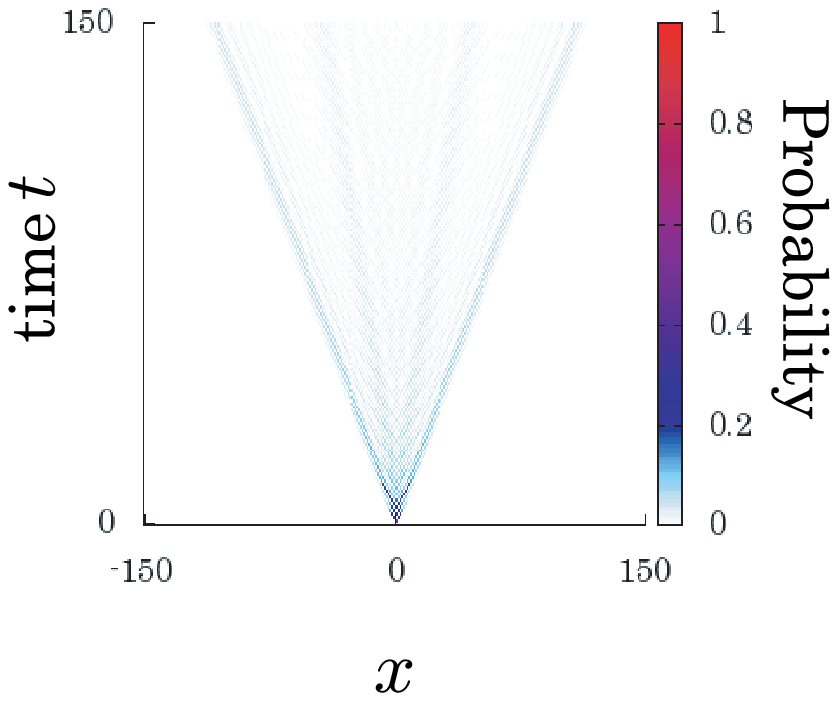}\\[2mm]
  (a) $\theta=\frac{\pi}{4}$
  \end{center}
 \end{minipage}
 \begin{minipage}{70mm}
  \begin{center}
   \includegraphics[scale=0.7]{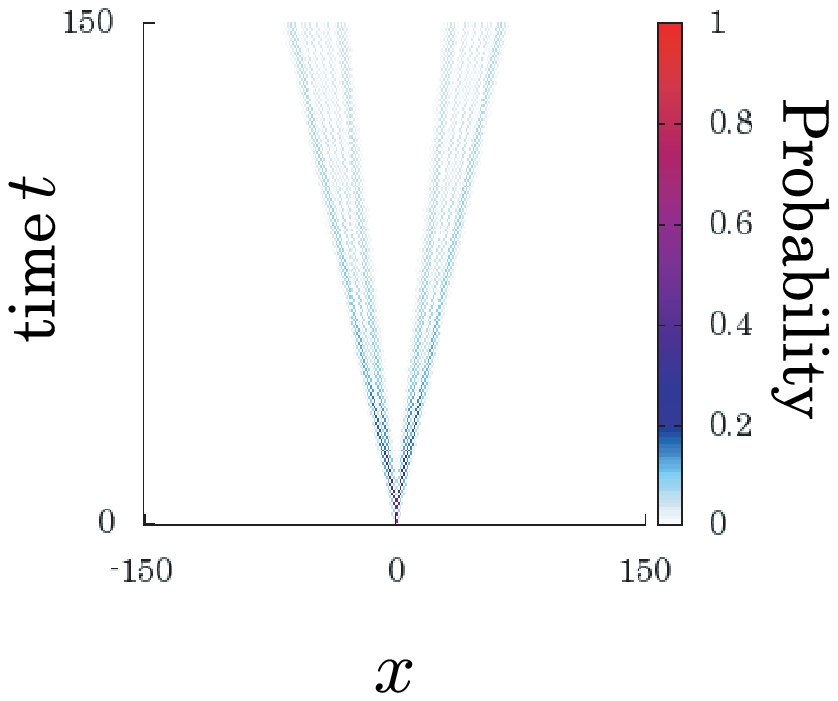}\\[2mm]
  (b) $\theta=\frac{2\pi}{5}$
  \end{center}
 \end{minipage}
\vspace{5mm}
\fcaption{Time evolution of probability distributions in the case of $\alpha=1/\sqrt{2}, \beta=i/\sqrt{2}$}
\label{fig:time-probability}
\end{center}
\end{figure}

\begin{figure}[h]
\begin{center}
 \begin{minipage}{70mm}
  \begin{center}
  \includegraphics[scale=0.7]{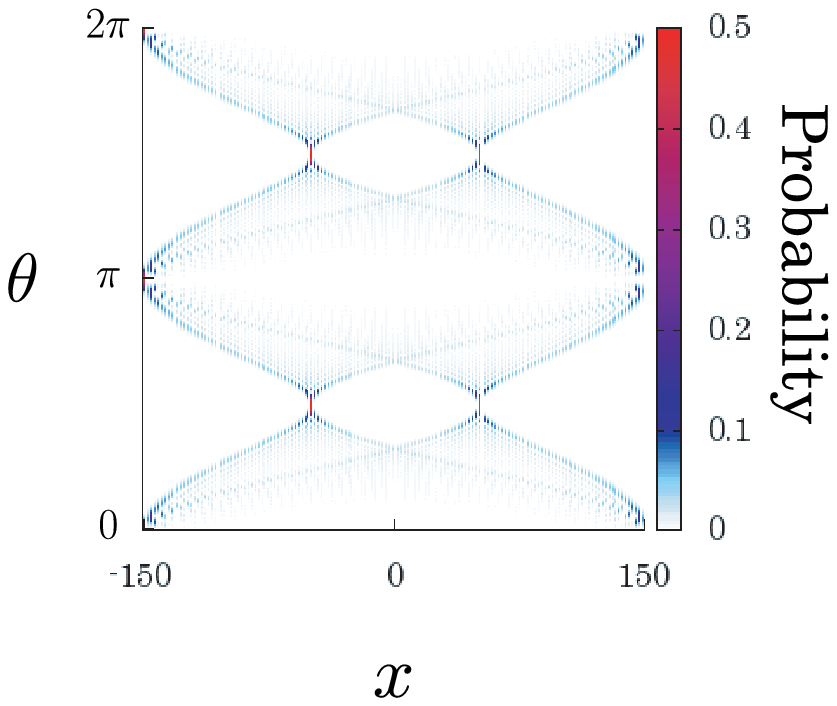}\\[2mm]
  (a) $\alpha=1/\sqrt{2},\,\beta=i/\sqrt{2}$
  \end{center}
 \end{minipage}
 \begin{minipage}{70mm}
  \begin{center}
   \includegraphics[scale=0.7]{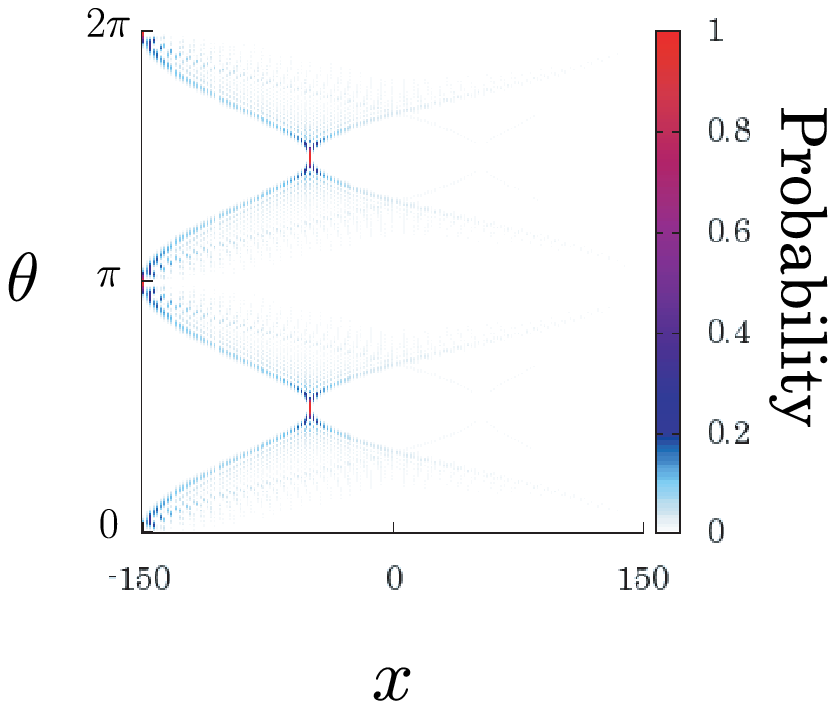}\\[2mm]
  (b) $\alpha=1,\,\beta=0$
  \end{center}
 \end{minipage}
\vspace{5mm}
\fcaption{The relationships between the probability distribution at time $t=150$ and the parameter $\theta$ which determines the coin-flip operator $C$}
\label{fig:theta-probability}
\end{center}
\end{figure}

%%%%%%%%%%%%%%%%%%%%%%%%%%%%%%%%%%%   LIMIT THEOREM   %%%%%%%%%%%%%%%%%%%%%%%%%%%%%%%%%%%%%%%%%%%%%%%%%%%%%%%%%
\section{Long-time limit theorem and its proof}
\label{sec:limit_th}

We get a long-time limit theorem for the probability distribution and its proof in this section assuming that the walker starts from the origin.
Let us take an initial state $\ket{\Psi_0}=\ket{0}\otimes\left(\alpha\ket{0}+\beta\ket{1}\right)$ with $|\alpha|^2+|\beta|^2=1$.
This initial condition means that the walker starts from the origin because of $\mathbb{P}(X_0=0)=1$. 
Then we obtain a limit theorem for the 3-period time-dependent QW.

\begin{thm}
 \begin{align}
   \lim_{t\to\infty}\mathbb{P}\left(\frac{X_{3t}}{3t}\leq x\right)
    =\int_{-\infty}^x \Biggl[&\left\{1-\nu(\alpha,\beta; y)\right\}f(y)I_{\left(\frac{1-4c^2}{3},\frac{\sqrt{1+8c^2}}{3}\right)}(y) \nonumber\\
    &+\left\{1+\nu(\alpha,\beta; -y)\right\}f(-y)I_{\left(-\frac{\sqrt{1+8c^2}}{3},-\frac{1-4c^2}{3}\right)}(y)\Biggr]\,dy,\label{eq:cumulative}
 \end{align}
 where
 \begin{align}
  f(x)=&\frac{|s|\left(|s|x+\sqrt{D(x)}\right)^2}{\pi(1-x^2)\sqrt{W_{+}(x)}\sqrt{W_{-}(x)}\sqrt{D(x)}},\\[3mm]
  \nu(\alpha,\beta; x)=&\frac{1}{c(1+8c^2)}\left\{9c^3(|\alpha|^2-|\beta|^2)+3s(1+6c^2)\Re(\alpha\overline{\beta})\right\}x\nonumber\\
  &+\frac{s}{c|s|(1+8c^2)}\left\{cs(|\alpha|^2-|\beta|^2)-(1+2c^2)\Re(\alpha\overline{\beta})\right\}\sqrt{D(x)},\\[3mm]
  D(x)=&1+8c^2-9c^2x^2,\\
  W_{+}(x)=&-(1-4c^2)+3(1-2c^2)x^2+2|s|x\sqrt{D(x)},\\
  W_{-}(x)=&1+8c^2-3(1+2c^2)x^2-2|s|x\sqrt{D(x)},\\[2mm]
  I_A(x)=&\left\{\begin{array}{cl}
	   1&(x\in A)\\
		  0&(x\notin A)
		 \end{array}\right.,
 \end{align}
 and $\Re(z)$ denotes the real part of the complex number $z$.
 \label{th:limit}
\end{thm}

The function $\nu(\alpha,\beta; x)$ is the part of the limit density function which gives the effect of the initial condition $\alpha,\beta$ on the limit behavior, and if the conditions $|\alpha|=|\beta|$ and $\Re(\alpha\overline{\beta})=0 $ are satisfied simultaneously (e.g. $\alpha=1/\sqrt{2}, \beta=i/\sqrt{2}$\,), this term disappears.
Note that $D(x), W_{+}(x), W_{-}(x) >0$ for $x\in \left(-\frac{\sqrt{1+8c^2}}{3},-\frac{1-4c^2}{3}\right) \cup \left(\frac{1-4c^2}{3},\frac{\sqrt{1+8c^2}}{3}\right)$ and $|1-4c^2|<\sqrt{1+8c^2}$.
As examples, Fig.~\ref{fig:limit} shows probability distributions and the limit density functions when $\alpha=1/\sqrt{2},\, \beta=i/\sqrt{2}$.
\begin{figure}[h]
\begin{center}
 \begin{minipage}{70mm}
  \begin{center}
   \includegraphics[scale=0.5]{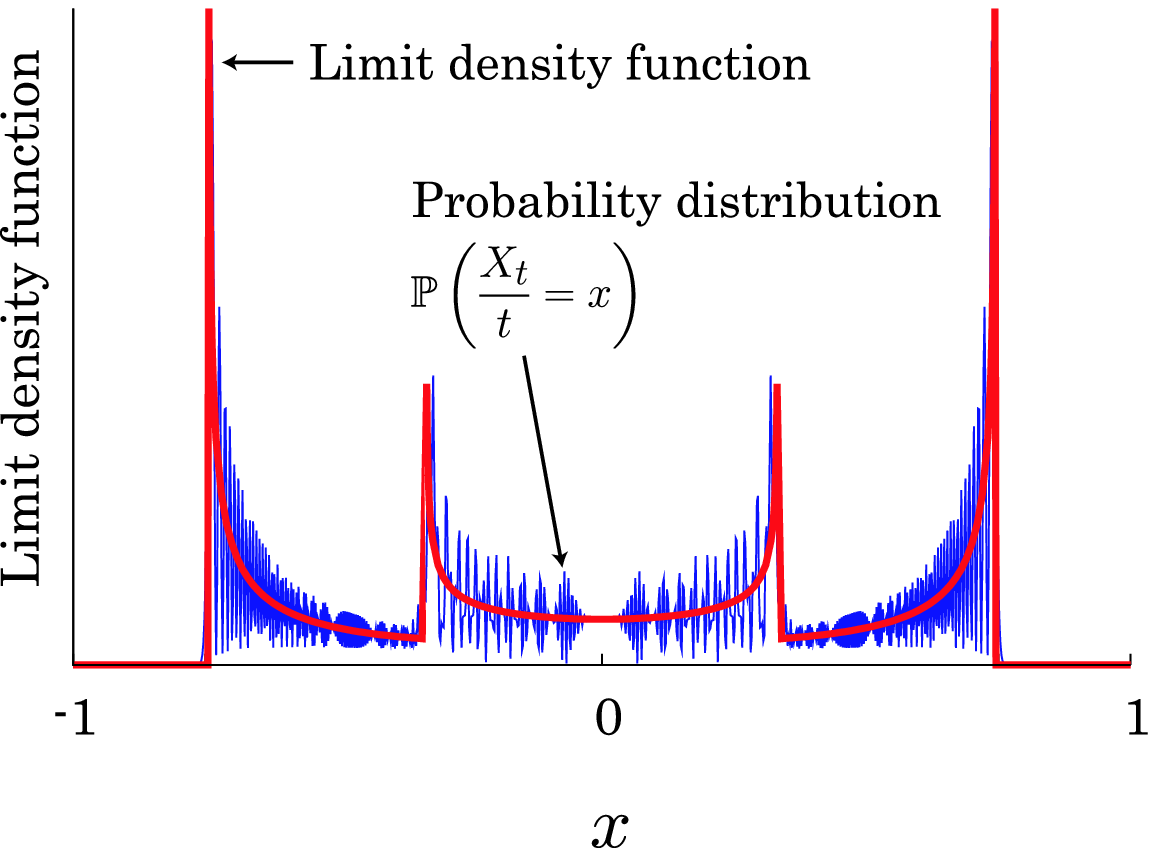}\\[2mm]
  (a) $\theta=\frac{\pi}{4}$
  \end{center}
 \end{minipage}
 \begin{minipage}{70mm}
  \begin{center}
   \includegraphics[scale=0.5]{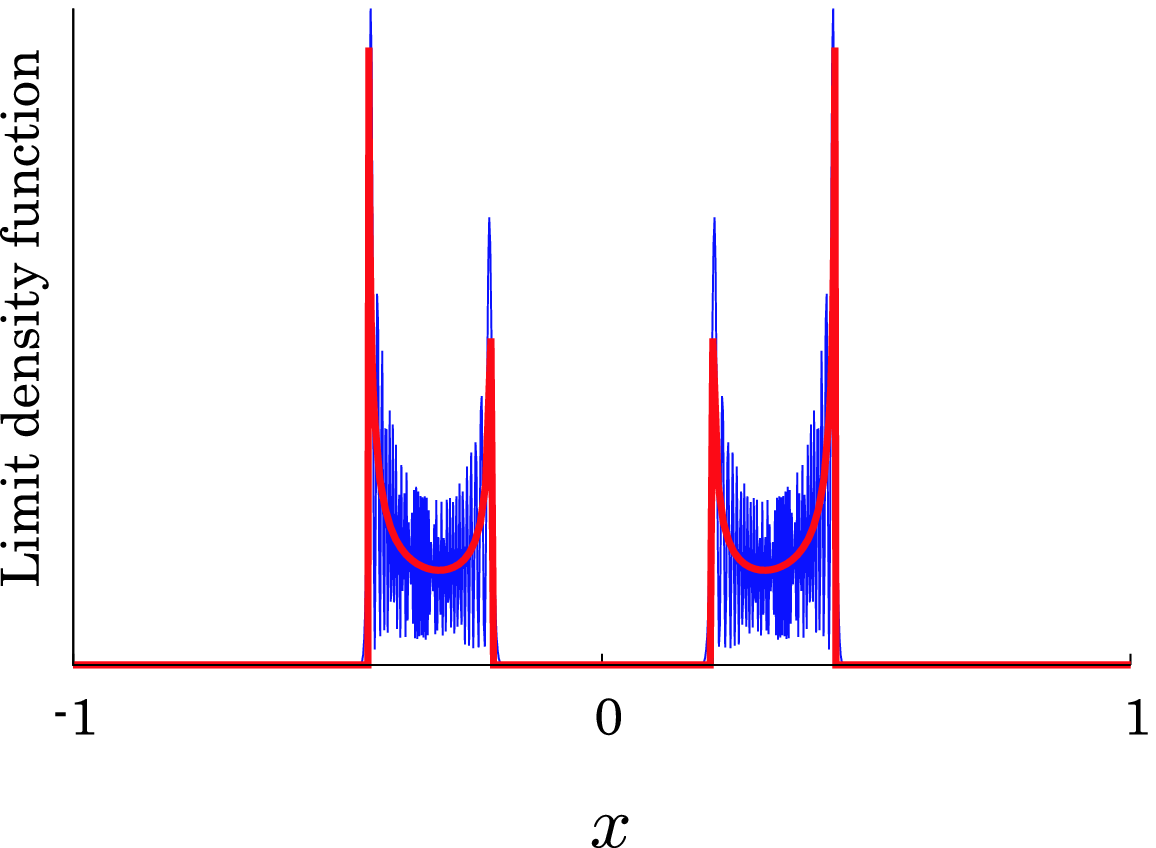}\\[2mm]
  (b) $\theta=\frac{2\pi}{5}$
  \end{center}
 \end{minipage}
\vspace{5mm}
\fcaption{Probability distribution at time $999\,(=3\times 333)$ (blue line) and the limit density function (red line), in the case of $\alpha=1/\sqrt{2}, \beta=i/\sqrt{2}$}
\label{fig:limit}
\end{center}
\end{figure}
\vspace{5mm}

\begin{proof}{
To prove the limit theorem we use Fourier analysis in the way introduced in Grimmett et al.~\cite{GrimmettJansonScudo2004}, and derive a convergence of the $r$-th moment $\mathbb{E}\left[(X_{3t}/3t)^r\right]$ ($r=0,1,2,\ldots$) which is equivalent to a convergence of the generating function $\mathbb{E}[e^{izX_{3t}/3t}]$.

First, we consider the following Fourier transform $\ket{\hat\Psi_t(k)}\, (k\in [-\pi,\pi))$ derived from the states of the walker
\begin{equation}
\ket{\hat\Psi_t(k)}=\sum_{x\in\mathbb{Z}}e^{-ikx}\ket{\psi_t(x)}.
\end{equation}
We should note that we can obtain the state $\ket{\psi_t(x)}$ by using the inverse Fourier transform
\begin{equation}
 \ket{\psi_t(x)}=\int_{-\pi}^\pi e^{ikx}\ket{\hat\Psi_t(k)}\frac{dk}{2\pi}.
\end{equation}
Equation~(\ref{eq:time-evolution}) produces a time evolution of the Fourier transform
\begin{align}
 \ket{\hat\Psi_{3t}(k)}=&\left(\hat S(k)\hat C(k)^2\right)^t\ket{\hat\Psi_0(k)},\nonumber\\
 \ket{\hat\Psi_{3t+1}(k)}=&\hat C(k)\left(\hat S(k)\hat C(k)^2\right)^t\ket{\hat\Psi_0(k)},\\
 \ket{\hat\Psi_{3t+2}(k)}=&\hat C(k)^2\left(\hat S(k)\hat C(k)^2\right)^t\ket{\hat\Psi_0(k)},\nonumber
\end{align}
where $\hat S(k)=e^{ik}\ket{0}\bra{0}+e^{-ik}\ket{1}\bra{1}$ and $\hat C(k)=\hat S(k)C$.
The operator $\hat S(k)$ corresponds to the position-shift operator $\tilde{S}$.

Before computing the $r$-th moment $\mathbb{E}\left(X_{3t}^r\right)$, we get the eigenvalues and the normalized eigenvectors of the unitary matrix $\hat S(k)\hat C(k)^2$ so that we rewrite the Fourier transform $\ket{\hat\Psi_{3t}(k)}$ on the appropriate eigenspace.
Let us take a standard basis as the orthogonal normalized basis $\left\{\ket{0},\ket{1}\right\}$ with
\begin{equation}
 \ket{0}=\left[\begin{array}{c}
	  1\\0
	       \end{array}\right],\quad
 \ket{1}=\left[\begin{array}{c}
	  0\\1
	       \end{array}\right].
\end{equation}
Then the matrix $\hat S(k)\hat C(k)^2$ has two eigenvalues
\begin{equation}
 \lambda_j(k)=c^2\cos 3k+s^2\cos k -(-1)^j i\sqrt{1-(c^2\cos 3k+s^2\cos k)^2} \quad (j=1,2),
\end{equation}
and they are distinct as long as $k\neq -\pi,0$.
Again, we should note that $1-(c^2\cos 3k+s^2\cos k)^2$ is not a negative number and its value is zero if and only if $k= -\pi,0$. 
As one of the possible expressions of the normalized eigenvector corresponding to each eigenvalue $\lambda_j(k)$, we have
\begin{equation}
 \ket{v_j(k)}=\frac{1}{\sqrt{N_j(k)}}\left[\begin{array}{c}
				      -2cs\,e^{2ik}\sin k\\[2mm]
					    c^2\sin 3k+s^2\sin k+(-1)^j\sqrt{1-(c^2\cos 3k+s^2\cos k)^2}
					   \end{array}\right],
\end{equation}
where $N_{j}(k)$ are normalization factors given by 
\begin{align}
 N_{j}(k)=&2\biggl\{1-(c^2\cos 3k+s^2\cos k)^2 \nonumber\\
 &+(-1)^j(c^2\sin 3k+s^2\sin k)\sqrt{1-(c^2\cos 3k+s^2\cos k)^2}\biggr\}.
\end{align}

Here, we treat the $r$-th moments at time $3t$ and express them in the Fourier space by using the eigenvalues $\lambda_j(k)$ and the eigenvectors $\ket{v_j(k)}$.
With a decomposition $\ket{\hat\Psi_{3t}(k)}=\sum_{j=0}^1\lambda_j^t(k)\braket{v_j(k)|\hat\Psi_0(k)}\ket{v_j(k)}$, we get
\begin{align}
 \mathbb{E}(X_{3t}^r)=&\sum_{x\in\mathbb{Z}}x^r\mathbb{P}(X_{3t}=x)\nonumber\\
 =&\int_{-\pi}^\pi \bra{\hat\Psi_{3t}(k)}\left(D^r\ket{\hat\Psi_{3t}(k)}\right)\frac{dk}{2\pi}\nonumber\\
 =&(t)_r\int_{-\pi}^\pi \sum_{j=1}^2 \left(\frac{i\lambda'_j(k)}{\lambda_j(k)}\right)^r\left|\braket{v_j(k)|\hat\Psi_0(k)}\right|^2\frac{dk}{2\pi}+O(t^{r-1}),
 \label{eq:r-th_moment}
\end{align}
where $D=i(d/dk)$ and $(t)_r=t(t-1)\times\cdots\times(t-r+1)$.
Equation~(\ref{eq:r-th_moment}) gives us a convergence as $t\to\infty$,
\begin{equation}
 \lim_{t\to\infty}\mathbb{E}\left[\left(\frac{X_{3t}}{3t}\right)^r\right]=\int_{-\pi}^\pi \sum_{j=1}^2 \left(\frac{i\lambda'_j(k)}{3\lambda_j(k)}\right)^r\left|\braket{v_j(k)|\hat\Psi_0(k)}\right|^2\frac{dk}{2\pi},
 \label{eq:r-th_moment2}
\end{equation}
where
\begin{equation}
 \frac{i\lambda'_j(k)}{3\lambda_j(k)}=(-1)^j \frac{3c^2\sin 3k+s^2\sin k}{3\sqrt{1-(c^2\cos 3k+s^2\cos k)^2}}.
\end{equation}
Setting $i\lambda'_j(k)/3\lambda_j(k)=x$ in Eq.~(\ref{eq:r-th_moment2}) takes us to our goal because we have
\begin{align}
 \lim_{t\to\infty}\mathbb{E}\left[\left(\frac{X_{3t}}{3t}\right)^r\right]=&\int_{-\infty}^\infty x^r \Biggl[\left\{1-\nu(\alpha,\beta; x)\right\}f(x)I_{\left(\frac{1-4c^2}{3},\frac{\sqrt{1+8c^2}}{3}\right)}(x) \nonumber\\
    &+\left\{1+\nu(\alpha,\beta; -x)\right\}f(-x)I_{\left(-\frac{\sqrt{1+8c^2}}{3},-\frac{1-4c^2}{3}\right)}(x)\Biggr]\,dx,
\end{align}
which means that the random variable $X_{3t}/3t$ converges in distribution to a random variable with a density function
\begin{equation}
 \left\{1-\nu(\alpha,\beta; x)\right\}f(x)I_{\left(\frac{1-4c^2}{3},\frac{\sqrt{1+8c^2}}{3}\right)}(x)+\left\{1+\nu(\alpha,\beta; -x)\right\}f(-x)I_{\left(-\frac{\sqrt{1+8c^2}}{3},-\frac{1-4c^2}{3}\right)}(x).
\end{equation}
To obtain the cumulative distribution function on the left hand side of Eq.~(\ref{eq:cumulative}), we need to integrate this density.
}
\end{proof}

%%%%%%%%%%%%%%%%%%%   SUMMARY   %%%%%%%%%%%%%%%%%%%%%%%%%%%%%%%%%%%%
\section{Summary and Discussion}
\label{sec:summary}

We have dealt with a 3-period time-dependent discrete-time 2-state QW on the line with the walker located at the origin at the initial time, and gave a limit theorem which gives the asymptotic behavior of the walker after a large number of steps.
On a rescaled space by time, the position of the walker converges in distribution to a random variable. The density function of the random variable has a compact support.
Its shape resembles that of a doubled arcsine distribution.
When we choose the parameter $\theta$, which determines the coin-flip operator $C$, in the open interval $(\pi/3, 2\pi/3) \cup (4\pi/3, 5\pi/3)$, we do not observe the walker at the starting point after a long time as shown in Fig.~\ref{fig:limit}-(b) because the compact support is the open interval $\left(-\frac{\sqrt{1+8\cos^2\theta}}{3},-\frac{1-4\cos^2\theta}{3}\right)\cup\left(\frac{1-4\cos^2\theta}{3},\frac{\sqrt{1+8\cos^2\theta}}{3}\right)$.
For a time-independent walk or a 2-period time-dependent walk starting from the origin, the initial condition at the origin produces a linear function in their limit density functions~\cite{Konno2002a,MachidaKonno2010}.
On the other hand, the 3-period time-dependent walk treated in this paper, features a non-linear term reflecting the initial condition at the origin, which is expressed by $\nu(\alpha,\beta; x)$ in the limit theorem, and the function $\sqrt{D(x)}=\sqrt{1+8c^2-9c^2x^2}$. 
We showed that the limit distribution of the 3-period time-dependent walk is essentially different from that of the time-independent walk or the 2-period time-dependent walk.

We have treated a 3-period time-dependent walk whose coin-state is flipped by only one coin-flip operator $C$ at time  
$t=0,1 \mod 3$, and is shifted without any coin-flip operation at time $t=2 \mod 3$.
We can also see very interesting behavior for a 3-period time-dependent walk with three distinct coin-flip operators by using numerics as displayed in the appendix. We intend to analyze these results carefully in a future publication.
We have described a mathematical property of a 3-period time-dependent walk. 
It would be worth discussing this phenomenon from the perspective of physics, for example it would be nice to explore a possible application to the design of selective pulses in~\cite{MorrisMcIntyreRourkeNgo1989}.

%%%%%%%%%%%%%%%%%ACKNOWLEDGEMENTS%%%%%%%%%%%%%%%%%%%%%%%%%%%%%%%%%%
\nonumsection{Acknowledgements}
\noindent T. Machida is grateful to the Japan Society for the Promotion of Science for the support, and to the Math. Dept. UC Berkeley for hospitality.
F.A. Gr\"{u}nbaum acknowledges support from the Applied Math. Sciences subprogram of the Office of Energy Research, US Department of Energy, under
Contract DE-AC03-76SF00098, and from AFOSR grant FA95501210087 through a subcontract to Carnegie Mellon University.

%%%%%%%%%%%%%%%%%%%   BIBTEX   %%%%%%%%%%%%%%%%%%%%%%%%%%%%%%%%%%

%%%%%%%%%%%%%%%%%%%%%%%% APPENDIX %%%%%%%%%%%%%%%%%%%%%%%%%%%%%%%%%%%%

\appendix{\quad A limit distribution for a general unitary operator}
Consider the following time evolution
\begin{equation}
 \ket{\Psi_{t+1}}=\left\{\begin{array}{ll}
		   \tilde{S}\tilde{U}\ket{\Psi_t}& (t=0,1 \mod 3)\\[1mm]
		    \tilde{S}\tilde{J}\ket{\Psi_t}& (t=2 \mod 3)
			 \end{array}\right.,
\end{equation}
where
\begin{align}
 \tilde{U}=&\sum_{x\in\mathbb{Z}}\ket{x}\bra{x}\otimes U,\\
 \tilde{J}=&\sum_{x\in\mathbb{Z}}\ket{x}\bra{x}\otimes J,\\
 U= & a\ket{0}\bra{0}+b\ket{0}\bra{1}+c\ket{1}\bra{0}+d\ket{1}\bra{1}\,\in U(2),\\
 J= & \ket{0}\bra{0}-\frac{\overline{a}d}{|a|^2}\ket{1}\bra{1}.
\end{align}
This differs from in Eq.~(\ref{eq:coin-flip operator}) in that we allow for complex entries in our coin.

\bigskip

\begin{lem}
We take the coin-flip operator which satisfies the condition $abcd\neq 0$.
If the walker starts with $\ket{\Psi_0}=\ket{0}\otimes (\alpha\ket{0}+\beta\ket{1})$, we have
 \begin{align}
   \lim_{t\to\infty}\mathbb{P}\left(\frac{X_{3t}}{3t}\leq x\right)
    =\int_{-\infty}^x \Biggl[&\left\{1-\chi(\alpha,\beta; y)\right\}f(y)I_{\left(\frac{1-4|a|^2}{3},\frac{\sqrt{1+8|a|^2}}{3}\right)}(y) \nonumber\\
    &+\left\{1+\chi(\alpha,\beta; -y)\right\}f(-y)I_{\left(-\frac{\sqrt{1+8|a|^2}}{3},-\frac{1-4|a|^2}{3}\right)}(y)\Biggr]\,dy,
 \end{align}
where
 \begin{align}
  f(x)=&\frac{|b|\left(|b|x+\sqrt{D(x)}\right)^2}{\pi(1-x^2)\sqrt{W_{+}(x)}\sqrt{W_{-}(x)}\sqrt{D(x)}},\\[3mm]
  \chi(\alpha,\beta; x)=&\frac{1}{|a|^2(1+8|a|^2)}\left\{9|a|^4(|\alpha|^2-|\beta|^2)+3(1+6|a|^2)\Re(a\alpha\overline{b\beta})\right\}x\nonumber\\
  &+\frac{1}{|a^2b|(1+8|a|^2)}\left\{|ab|^2(|\alpha|^2-|\beta|^2)-(1+2|a|^2)\Re(a\alpha\overline{b\beta})\right\}\sqrt{D(x)},\\[3mm]
  D(x)=&1+8|a|^2-9|a|^2x^2,\\
  W_{+}(x)=&-(1-4|a|^2)+3(1-2|a|^2)x^2+2|b|x\sqrt{D(x)},\\
  W_{-}(x)=&1+8|a|^2-3(1+2|a|^2)x^2-2|b|x\sqrt{D(x)}.
 \end{align}
 \label{app:th:limit}
\end{lem}

\begin{proof}{
We parametrize the unitary operator $U$ by the choices $a=e^{i(\gamma+\xi)}\cos\theta, b=e^{i(\gamma-\xi)}\sin\theta, c=e^{i(\delta+\xi)}\sin\theta, d=-e^{i(\delta-\xi)}\cos\theta\,(\gamma,\delta,\xi,\theta\in [0,2\pi))$ and get
\begin{equation}
 \hat{S}(k)J\left(\hat{S}(k)U\right)^2=e^{i\left(\frac{\gamma}{2}+\frac{3\delta}{2}-\xi\right)}\hat{S}(-\xi)\hat{S}\left(k+\frac{\gamma-\delta+2\xi}{2}\right)\left\{\hat{S}\left(k+\frac{\gamma-\delta+2\xi}{2}\right)C\right\}^2\hat{S}(\xi).
\end{equation}
From Theorem~\ref{th:limit}, hence, we have
\begin{align}
 \lim_{t\to\infty}\mathbb{P}\left(\frac{X_{3t}}{3t}\leq x\right)
 =&\int_{-\infty}^x \Biggl[\left\{1-\nu(\alpha e^{i\xi},\beta e^{-i\xi}; y)\right\}f(y)I_{\left(\frac{1-4\cos^2\theta}{3},\frac{\sqrt{1+8\cos^2\theta}}{3}\right)}(y) \nonumber\\
 &+\left\{1+\nu(\alpha e^{i\xi},\beta e^{-i\xi}; -y)\right\}f(-y)I_{\left(-\frac{\sqrt{1+8\cos^2\theta}}{3},-\frac{1-4\cos^2\theta}{3}\right)}(y)\Biggr]\,dy,
\end{align}
for $\theta\neq 0,\frac{\pi}{2},\pi,\frac{3\pi}{2}$.
}
\end{proof}

\clearpage

\appendix{\quad Using three different flip coins}

The following figures show the results of some numerical simulations using three different coins.
We focus on a general 3-period time evolution
\begin{equation}
 \ket{\Psi_{t+1}}=\left\{\begin{array}{ll}
		   \tilde{S}\tilde{U}_0\ket{\Psi_t}& (t=0\mod 3)\\[1mm]
		   \tilde{S}\tilde{U}_1\ket{\Psi_t}& (t=1\mod 3)\\[1mm]
		   \tilde{S}\tilde{U}_2\ket{\Psi_t}& (t=2\mod 3)
			 \end{array}\right.,
\end{equation}
where
\begin{equation}
 U_j=\left[\begin{array}{cc}
      e^{i\gamma_j}&0 \\
	    0&e^{i\delta_j}
	   \end{array}\right]
 \left[\begin{array}{cc}
  \cos\theta_j & \sin\theta_j\\
	\sin\theta_j & -\cos\theta_j
       \end{array}\right]
 \left[\begin{array}{cc}
  e^{i\xi_j}&0 \\
	0& e^{-i\xi_j}
       \end{array}\right]\quad(j=0,1,2).
\end{equation}
\begin{figure}[h]
\begin{center}
 \begin{minipage}{70mm}
  \begin{center}
   \includegraphics[scale=0.45]{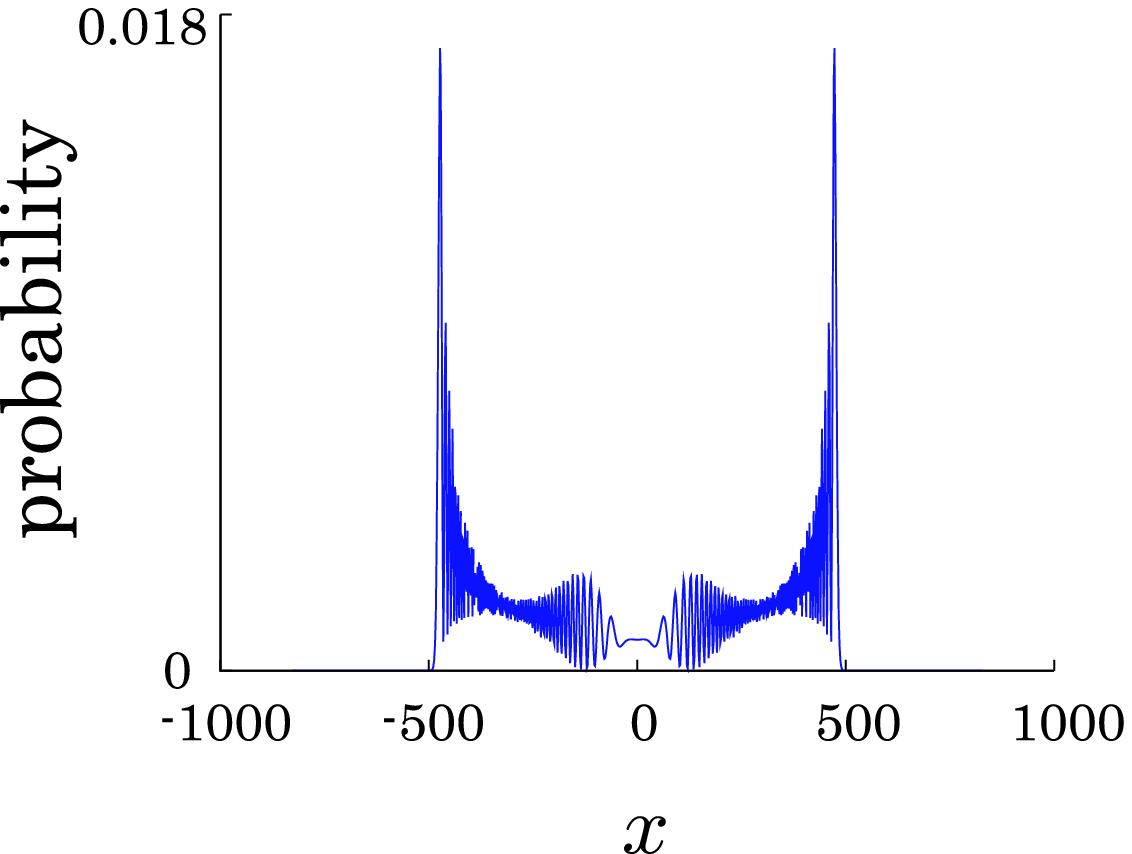}\\[2mm]
  (a)\\[1mm] $(\gamma_0,\delta_0,\xi_0)=(0,0,0)$\\$(\gamma_1,\delta_1,\xi_1)=(0,0,0)$\\$(\gamma_2,\delta_2,\xi_2)=(0,0,0)$
  \end{center}
 \end{minipage}
 \begin{minipage}{70mm}
  \begin{center}
   \includegraphics[scale=0.45]{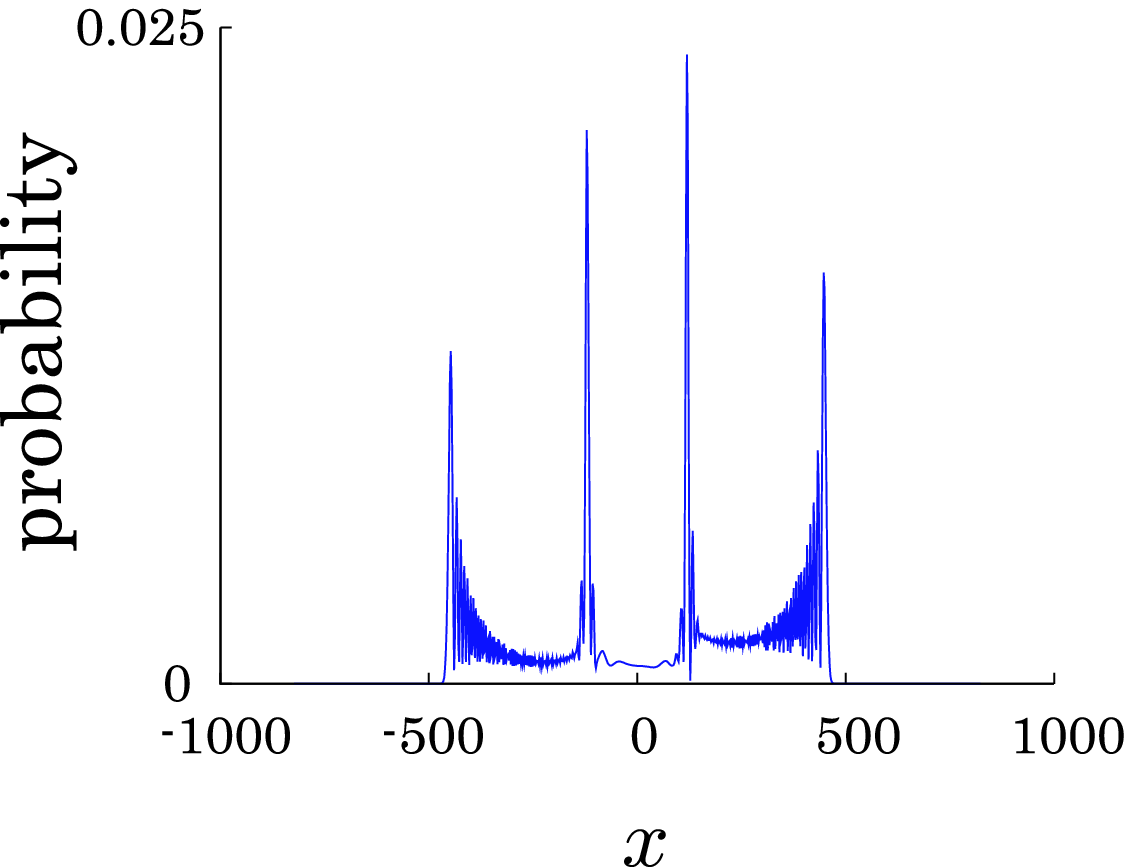}\\[2mm]
  (b)\\[1mm] $(\gamma_0,\delta_0,\xi_0)=(\frac{\pi}{4},0,0)$\\$(\gamma_1,\delta_1,\xi_1)=(0,0,0)$\\$(\gamma_2,\delta_2,\xi_2)=(0,0,0)$
  \end{center}
 \end{minipage}

 \bigskip

 \begin{minipage}{70mm}
  \begin{center}
   \includegraphics[scale=0.45]{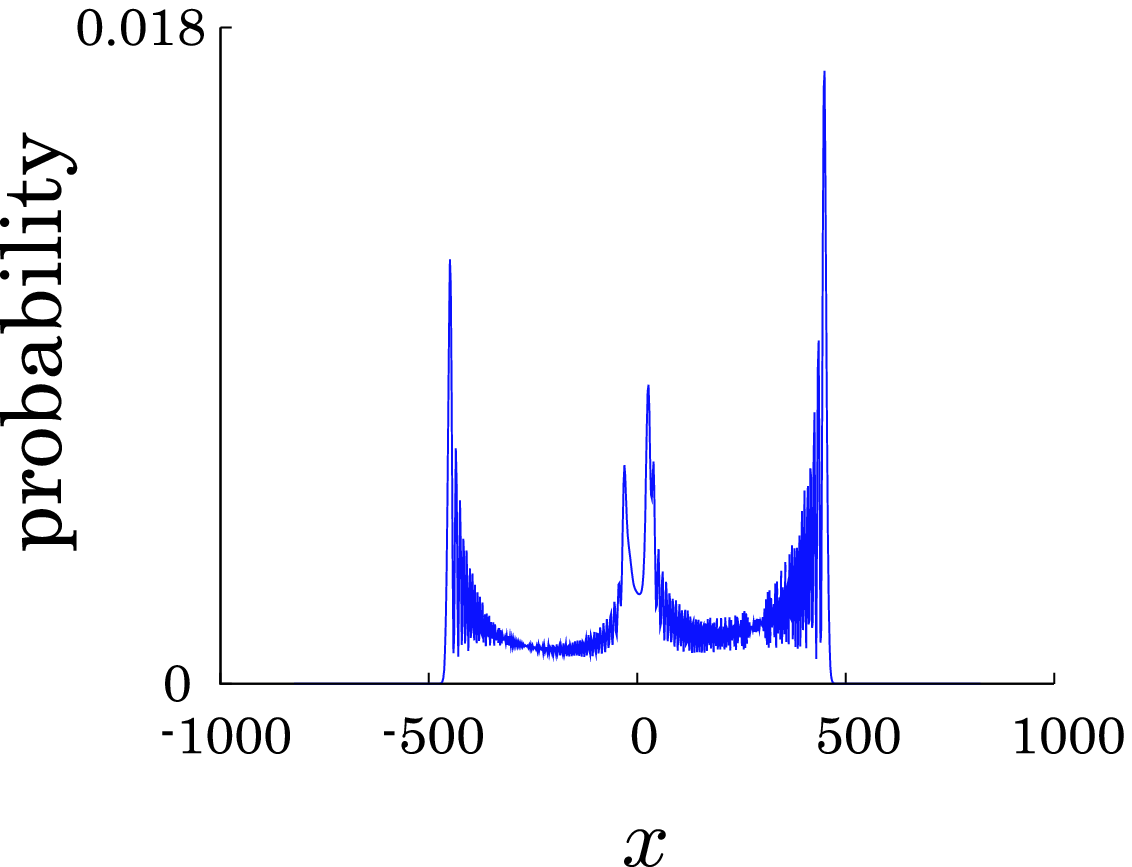}\\[2mm]
  (c)\\[1mm] $(\gamma_0,\delta_0,\xi_0)=(0,0,0)$\\$(\gamma_1,\delta_1,\xi_1)=(0,0,0)$\\$(\gamma_2,\delta_2,\xi_2)=(\frac{\pi}{4},0,0)$
  \end{center}
 \end{minipage}
 \begin{minipage}{70mm}
  \begin{center}
   \includegraphics[scale=0.45]{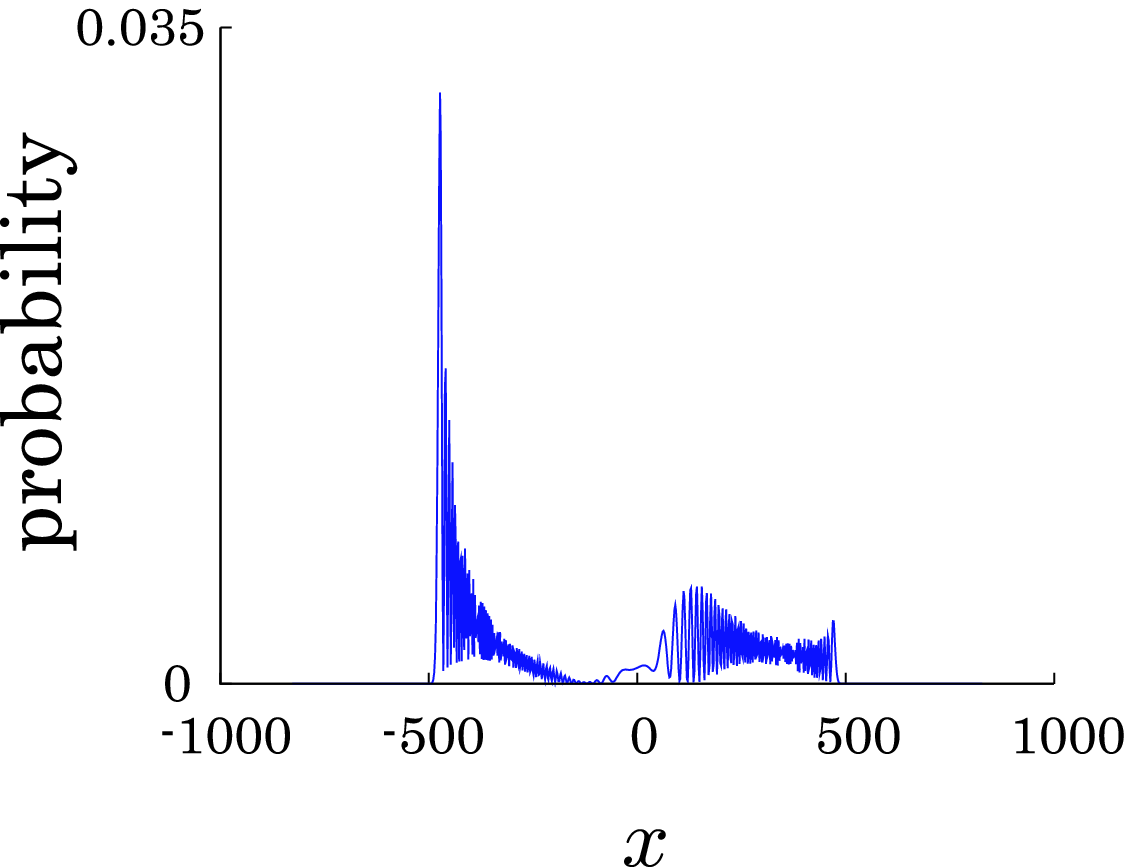}\\[2mm]
  (d)\\[1mm] $(\gamma_0,\delta_0,\xi_0)=(\frac{\pi}{2},\frac{\pi}{2},\frac{\pi}{2})$\\$(\gamma_1,\delta_1,\xi_1)=(\frac{\pi}{3},\frac{\pi}{3},\frac{\pi}{3})$\\$(\gamma_2,\delta_2,\xi_2)=(\frac{\pi}{4},\frac{\pi}{4},\frac{\pi}{4})$
  \end{center}
 \end{minipage}
\vspace{5mm}
\fcaption{Probability distributions at time $t=999$ in the case of $\theta_1=\frac{2\pi}{5}, \theta_2=\frac{\pi}{3}, \theta_3=\frac{\pi}{4}, \alpha=1/\sqrt{2}, \beta=i/\sqrt{2}$}
\end{center}
\end{figure}

\clearpage

\appendix{\quad Probability distributions at time $3t+1,\, 3t+2$ and the limit density functions}
\label{app:3t+1_and_3t+2}

We show comparisons between the probability distributions of the rescaled random valuables $\frac{X_{3t+1}}{3t+1},\,\frac{X_{3t+2}}{3t+2}$ and the limit density functions which follow from Theorem~\ref{th:limit}.
\bigskip

\begin{figure}[h]
\begin{center}
 \begin{minipage}{70mm}
  \begin{center}
   \includegraphics[scale=0.5]{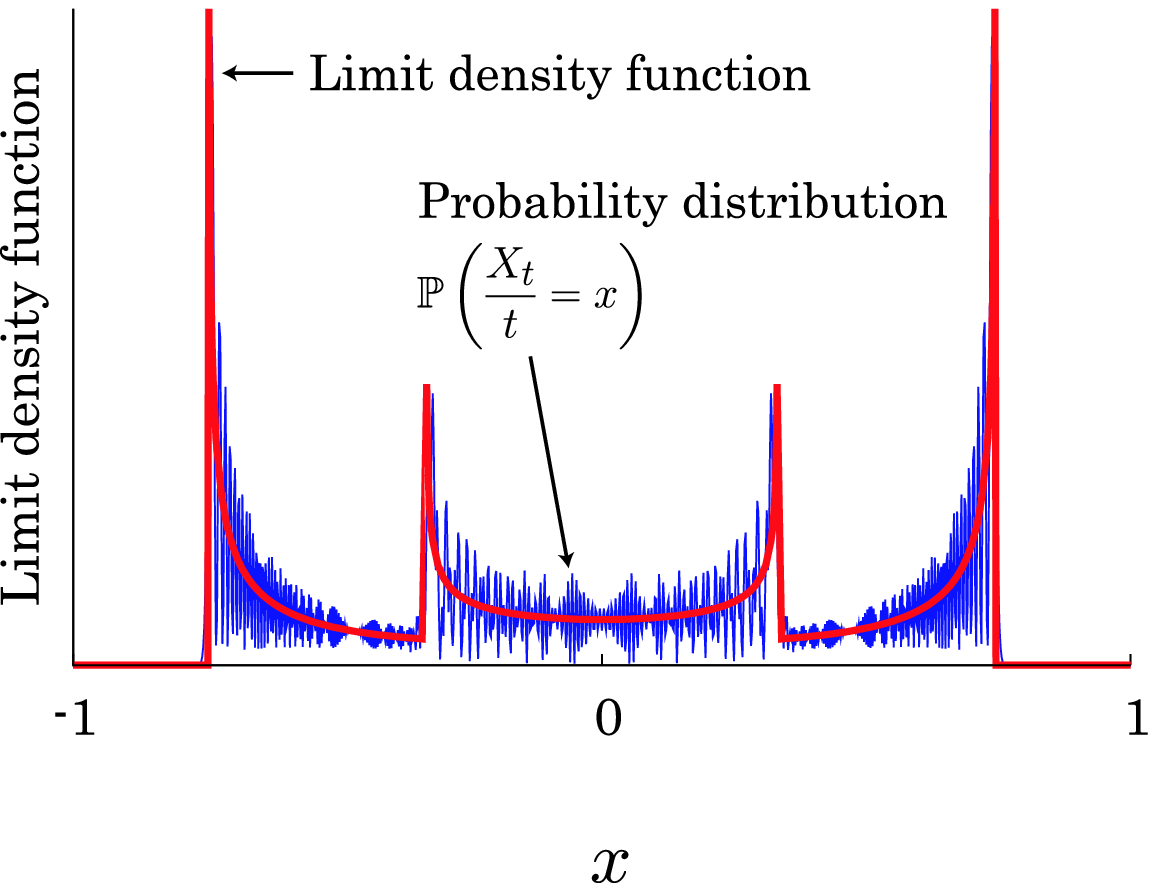}\\[2mm]
  (a) $\theta=\frac{\pi}{4}$
  \end{center}
 \end{minipage}
 \begin{minipage}{70mm}
  \begin{center}
   \includegraphics[scale=0.5]{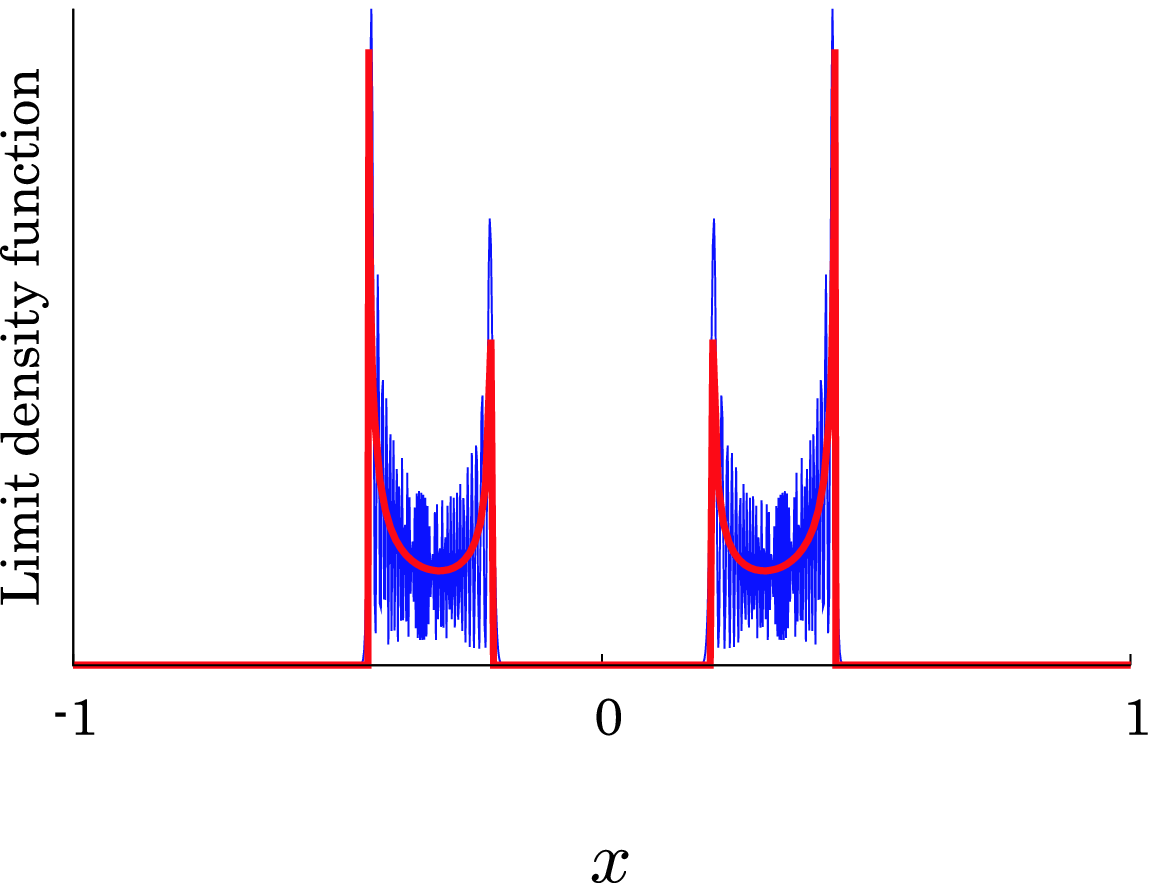}\\[2mm]
  (b) $\theta=\frac{2\pi}{5}$
  \end{center}
 \end{minipage}
\vspace{5mm}
\fcaption{Probability distribution at time $1000\,(=3\times 333+1)$ (blue line) and the limit density function (red line), in the case of $\alpha=1/\sqrt{2}, \beta=i/\sqrt{2}$}
\end{center}
\end{figure}

\begin{figure}[h]
\begin{center}
 \begin{minipage}{70mm}
  \begin{center}
   \includegraphics[scale=0.5]{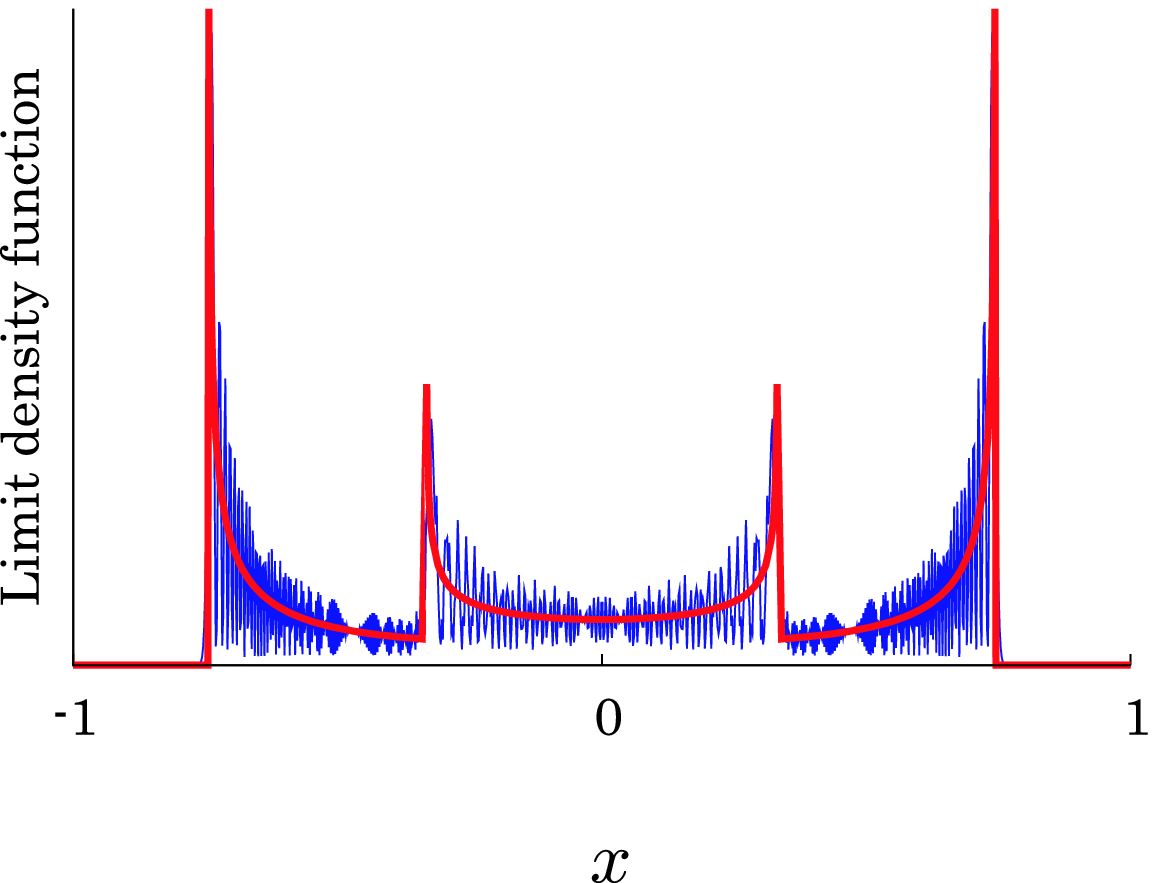}\\[2mm]
  (a) $\theta=\frac{\pi}{4}$
  \end{center}
 \end{minipage}
 \begin{minipage}{70mm}
  \begin{center}
   \includegraphics[scale=0.5]{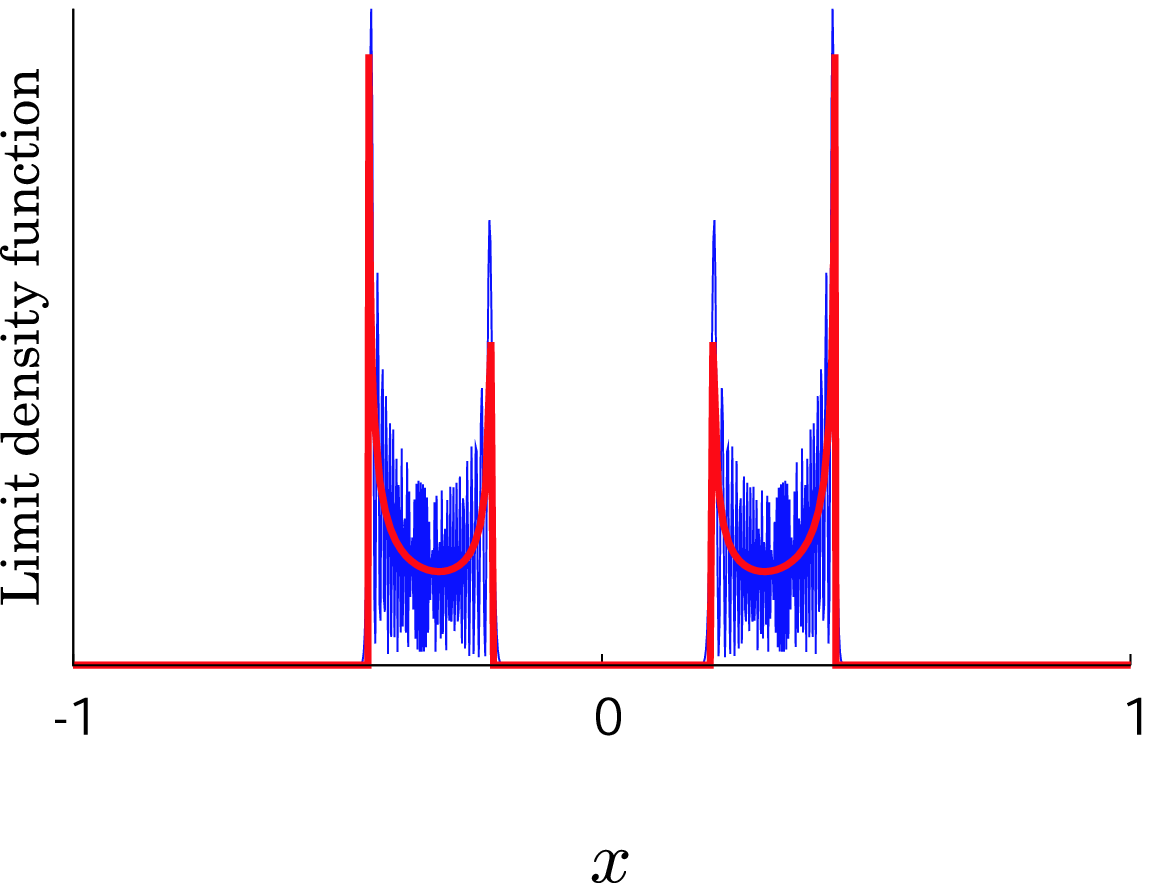}\\[2mm]
  (b) $\theta=\frac{2\pi}{5}$
  \end{center}
 \end{minipage}
\vspace{5mm}
\fcaption{Probability distribution at time $1001\,(=3\times 333+2)$ (blue line) and the limit density function (red line), in the case of $\alpha=1/\sqrt{2}, \beta=i/\sqrt{2}$}
\end{center}
\end{figure}

\end{document}